\documentclass[runningheads]{llncs}
\usepackage[utf8]{inputenc}
\usepackage[T1]{fontenc}
\usepackage{amsmath,amssymb}
\usepackage{xspace}
\usepackage{enumitem}
\usepackage[usenames,dvipsnames,table]{xcolor}
\usepackage[font=small, labelfont=bf]{caption}
\usepackage[font=small, labelfont=small]{subcaption}
\usepackage{wrapfig,graphicx}
\DeclareGraphicsExtensions{.pdf,.png,.eps}
\graphicspath{{figs/}}
\definecolor{darkblue}{rgb}{0.121,0.47,0.705}
\newcommand{\bl}{\color{darkblue}}
\definecolor{dark purple}{rgb}{0.415,0.239,0.603}
\newcommand{\pu}{\color{dark purple}}
\definecolor{dark orange}{rgb}{1,0.498,0}

\let\emph\relax
\DeclareTextFontCommand{\emph}{\bl\em}
\let\doendproof\endproof\renewcommand\endproof{~\hfill$\qed$\doendproof}
\usepackage{thm-restate}

\usepackage{mathtools}
\DeclarePairedDelimiter\set{\{}{\}}
\DeclarePairedDelimiter\abs{\lvert}{\rvert}

\DeclarePairedDelimiter\croc{\langle}{\rangle}
\def\Oh{\ensuremath{\mathcal{O}}}
\newcommand{\VTT}{\textsc{VTT}\xspace}
\newcommand{\FTT}{\textsc{FTT}\xspace}

\usepackage{url}
\usepackage{cite}
\usepackage[hidelinks=true,
			bookmarks=true,
			bookmarksopen=true,
			bookmarksopenlevel=2,
			bookmarksnumbered=true,
			hyperindex=true,
			plainpages=false,
			pdfpagelabels=true
]{hyperref} 
\usepackage[capitalise,nameinlink]{cleveref}
\crefname{observation}{Observation}{Observations}
\newcommand{\etal}{{et~al.}}

\renewcommand{\orcidID}[1]{\href{https://orcid.org/#1}{\includegraphics[scale=.03]{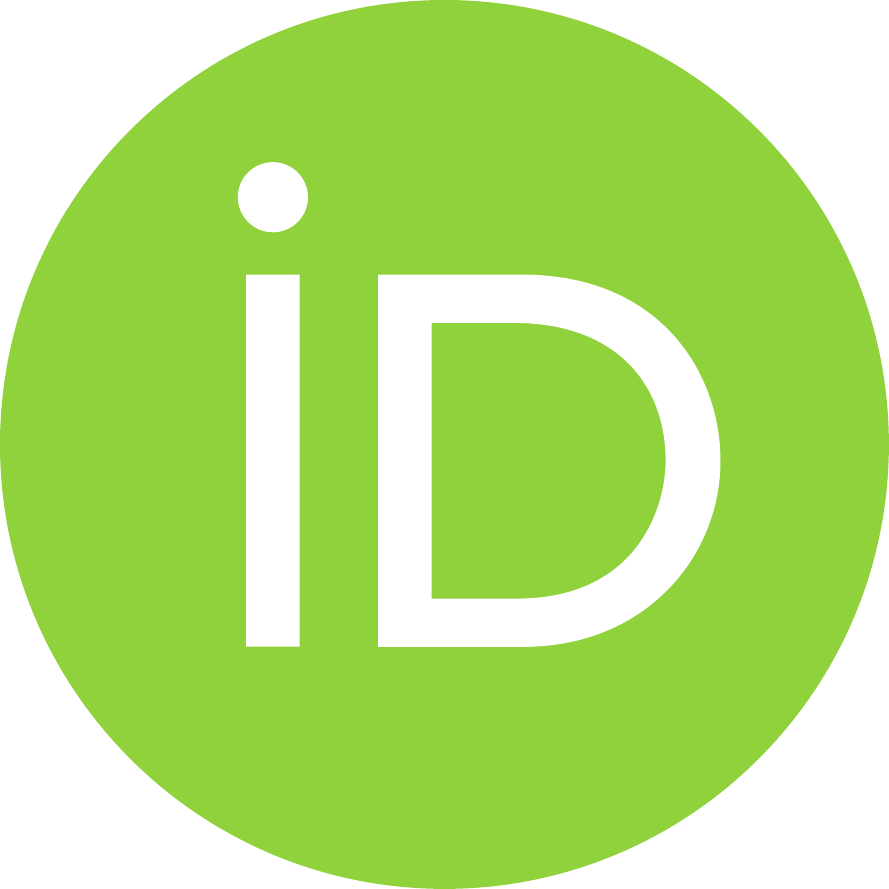}}}

\title{Visualizing Multispecies Coalescent~Trees:
  Drawing~Gene Trees Inside Species Trees}

\titlerunning{Visualizing Multispecies Coalescent Trees}

\author{Jonathan~Klawitter\inst{1,2}\orcidID{0000-0001-8917-5269} 
\and Felix~Klesen\inst{1}\orcidID{0000-0003-1136-5673} 
\and Moritz~Niederer\inst{3} 
\and Alexander~Wolff\inst{1}\orcidID{0000-0001-5872-718X}}
\authorrunning{J.~Klawitter et al.}
\institute{Universität Würzburg, Germany 
\and University of Auckland, New Zealand 
\and HTW Saar, Germany}

\begin{document}

\maketitle        

\pdfbookmark[1]{Abstract}{Abstract} 
\begin{abstract}
  We consider the problem of drawing multiple gene trees inside a
  single species tree in order to visualize multispecies coalescent
  trees.  Specifically, the drawing of the species tree fills a
  rectangle in which each of its edges is represented by a smaller
  rectangle, and the gene trees are drawn as rectangular cladograms
  (that is, orthogonally and downward, with one bend per edge) inside
  the drawing of the species tree.  As an alternative, we also
  consider a style where the widths of the edges of the species tree
  are proportional to given effective population sizes.
  
  In order to obtain readable visualizations, our aim is to minimize
  the number of crossings between edges of the gene trees in such
  drawings.  We show that planar instances can be recognized in
  linear time and that the general problem is NP-hard.
  Therefore, we introduce
  two heuristics and give an integer linear programming (ILP)
  formulation that provides us with exact solutions in exponential
  time.  We use the ILP to measure the quality of the heuristics on
  real-world instances.  The heuristics yield surprisingly good
  solutions, and the ILP runs surprisingly fast.
\end{abstract}

\section{Introduction}
Visualizations of trees to present information have been used for centuries~\cite{Lim14} and 
the study of producing readable, compact representations of trees has a long tradition~\cite{Treevis,Rus13}. 
Trees and their drawings are also an ubiquitous and fundamental tool in the field of phylogenetics. 
In particular, a phylogenetic tree is used
\begin{wrapfigure}[11]{r}{5cm}
	\vspace{-0.65cm}
	\centering
	\includegraphics[page=1]{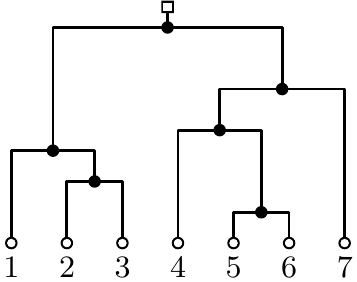}
	\caption{A rectangular cladogram drawing of a rooted binary phylogenetic tree on seven taxa.}
	\label{fig:phyTreeExamples:cladogram}
\end{wrapfigure}
to model the evolutionary history
and relationships of a set $X$ of taxa such as species, genes, or languages~\cite{SS03}.
There exist many different models, but most commonly a \emph{phylogenetic tree} on $X$ is a tree $T$
whose leaves are bijectively labeled with~$X$;
see \cref{fig:phyTreeExamples:cladogram}.
In a {\em rooted} phylogenetic tree, 
each internal vertex represents a branching event (such as species divergence);
time (or genetic distance) is represented by the edge lengths\pagebreak[1] 
from the root towards the leaves. 
In most\pagebreak[4] applications, the tree is \emph{binary}, that is,
each internal vertex has indegree one, outdegree two, and thus represents a bifurcation event.
An {\em unrooted} phylogenetic tree, on the other hand, models only the relatedness of the taxa.
A phylogenetic tree where the taxa are species is called a \emph{species tree}.
If the taxa are biological sequences, such as particular genes or protein sequences, the tree is called a \emph{gene~tree}. 

\paragraph{Multispecies coalescent models.}
One of the main tasks in phylogenetics is the inference of 
a phylogenetic tree for some given data and model.
When inferring a species tree based on sequencing data, one might be inclined 
to set the species tree as that of an inferred gene tree.
However, gene trees can differ from the species tree 
in the presence of so-called \textit{incomplete lineage sorting}\footnote{We
speak of \emph{incomplete lineage sorting} if 
(i)~in a population of an ancestral species
two (or more) variants of a gene were present, say red and blue, and 
(ii)~when the species diverged, this did not result in one child species 
having the red variant and the other having the blue variant,
but, e.g., one child species having both~variants~\cite{SS20}.} 
or when divergence times are small\footnote{A 
small divergence time corresponds to a short edge in the phylogenetic tree,
which can be hard to infer correctly.},
which can lead to inaccurate edge lengths or even to an incorrectly inferred species tree~\cite{PN88,AEWBS02,MH16,SS20}.
To address these issues, \emph{multispecies coalescent (MSC) models} have been developed.
An MSC model provides a framework for inferring species trees 
while accounting for conflicts between gene trees and species trees~\cite{HD09,FJRY18,RELY20}.
Roughly speaking, by using multiple samples (genes) per species, 
the model coestimates multiple gene trees that are constrained
within their shared species tree.
In doing so, the model can infer not only divergence times for inner vertices
but also the {\em effective population size} for each edge (\emph{branch}) in the species tree.
There exist several models for population sizes~\cite{WWB03}, two of which we define here. 
In the \emph{continuous linear model}, for each branch,
the population size between the top and the bottom is linearly interpolated,
and for a branch not incident to a leaf,
the population size at the bottom equals 
the sum of the population sizes at the top of its two child branches; see~\cref{fig:mscExamples:continuous}.
In the \emph{piecewise constant model}, the population size of each branch
is constant from the top to the bottom of the branch and 
there are no restrictions between adjacent branches~\cite{Dou20}; see~\cref{fig:mscExamples:constant}.

\begin{figure}[htb]
  \centering
    \begin{subfigure}[t]{0.45 \linewidth}
		\centering
 		\includegraphics[page=2]{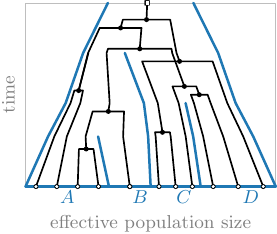}
		\caption{Continuous linear model.}
		\label{fig:mscExamples:continuous}
	\end{subfigure}
	\hfill
	\begin{subfigure}[t]{0.45 \linewidth}
		\centering
 		\includegraphics[page=3]{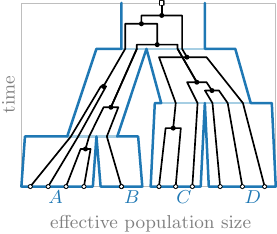}
		\caption{Piecewise constant model.}
		\label{fig:mscExamples:constant}
	\end{subfigure}
  \caption{Multispecies coalescent of a {\bl species tree} on four species $A, B, C, D$ 
  and a gene tree on eleven taxa under two different models.}
  \label{fig:mscExamples}
\end{figure}

For a phylogenetic tree $T$, let $V(T)$ be the vertex set of $T$,
let $E(T)$ be the edge set of $T$, and let $L(T)$ be the leaf set of $T$.
We define an \emph{MSC tree} as a triple~$\croc{S, T, \varphi}$ 
consisting of a species tree $S$, a gene tree $T$, 
and a mapping~$\varphi \colon L(T) \to L(S)$ with the following properties.
Both $S$ and $T$ are rooted binary phylogenetic trees
where all vertices have an associated height $h$
that are strictly decreasing from root to leaf.
We consider only the case where $h(\ell)$ is zero for each leaf $\ell$ in $L(S)$ and $L(T)$. 
For gene trees, we use the terms \emph{leaf}, \emph{vertex}, and \emph{edge};
whereas we use the terms \emph{species}, \emph{node}, and
\emph{branch} if we want to stress that we talk about species trees.
Each branch in $E(S)$ is associated with an upper and a lower population size.
The mapping~$\varphi$ describes which leaves of~$T$ belong to which species in $S$.
Next, consider two leaves~$\ell$ and~$\ell'$ of~$T$ with $\varphi(\ell) \neq \varphi(\ell')$.
Let $v$ be the lowest common ancestor of $\ell$ and $\ell'$.
In the MSC model we have that a divergence event of $v$ occurred before
the divergence event at a node $s$ of $S$ that ultimately split $\varphi(\ell)$ and $\varphi(\ell')$.
Hence, $h(v) > h(s)$ and we can extrapolate~$\varphi$ 
to a mapping of each inner vertex $v$ of $T$ to a branch of~$S$.
Lastly, we assume that the input consists of a single gene tree;
otherwise we merge multiple given gene trees 
by connecting all their roots to a super root.

\paragraph{Visualizing MSC trees.}
Visualizations of an MSC tree usually show the species and gene tree together.
This allows the user to detect any discordance between them
such as whether they have different topologies and where incomplete lineage sorting occurs.
It is also interesting to see where these events occur with respect to the inferred population sizes.
Furthermore, these drawings are used to diagnose whether the parameters of the model are set up well.   
E.g., if all inner vertices of the gene tree occur directly above nodes of the species tree
or if all occur near the root of the species tree, parameters may have been chosen poorly. 

Wilson \etal~\cite{WWB03} suggested a {\em tree-in-tree} style 
for an MSC tree $\croc{S, T, \varphi}$ under continuous models similar to the one shown in \cref{fig:mscExamples:continuous}.
There, the species tree~$S$ is drawn in a space-filling fashion
such that the branch widths of $S$ reflect the associated population sizes
and the gene tree $T$ is then drawn into $S$ as a classic node-link diagram.
Without these constraints on the branch widths,
$T$ could be drawn as a classic rectangular cladogram as in \cref{fig:phyTreeExamples:cladogram}.
As noted above, the MSC model ensures that~$T$ can be drawn inside $S$ without edges
of~$T$ crossing edges of~$S$
since, for each edge~$uv$ of $T$, we have that, if $u$ and $v$ lie inside 
the branches~$e_u$ and~$e_v$ of~$S$, respectively,
then either $e_u$ precedes $e_v$ or $e_u = e_v$.

Douglas~\cite{Dou20} developed the tool \texttt{UglyTrees} 
that generates tree-in-tree drawings for MSC trees under the piecewise constant model;
\cref{fig:mscExamples:constant} resembles such~a drawing.
He points out that the results are in many cases visually unpleasing 
	(as reflected in his choice for the tool's name),
in particular if the difference in width between parent and child is large.
This is amplified in practice by the inverse relationship between
the number of gene tree vertices and population sizes, 
which results in clusters of vertices in the narrowest branches~\cite{Dou20}.

\paragraph{Related work.}
There exist several applications where multiple phylogenetic trees are displayed together.
In a \emph{tanglegram}, two phylogenetic trees on the same set of taxa
are drawn opposite each other and the corresponding leaves are
connected by line segments for easy comparison~\cite{FKP10,BBBNOSW12}.
The tool \texttt{DensiTree}~\cite{Bou10} allows the user to compare
many trees simultaneously by drawing them on top of each other.
A \emph{co-phylogenetic tree} consists of two rooted phylogenetic trees, 
namely, a \emph{host tree} $H$ and a \emph{parasite tree} $P$,
together with a mapping (\emph{reconciliation}) of the vertices of $P$ to vertices of $H$.
Other than in an MSC tree, the vertices of $P$ commonly do not have heights 
but are mapped to nodes of $H$,
the host branches do not have associated population sizes,
and the edges of $P$ can go from one subtree of~$H$ to another,
representing so-called {\em host switches}.
Several tools visualize the reconciliation of co-phylogenetic trees~\cite{MM05,SSSLA07,CFOLH10,CDSJJB16}.
Commonly, the branches of $H$ are drawn with thick lines 
such that $P$ can be embedded into~$H$; see~\cref{fig:phyTreeExamples:reconciliation}.
Recently, Calamoneri \etal~\cite{CDMP20} suggested a tree-in-tree style for reconciliation
similar to the one for MSC trees above.
They draw~$H$ in a space-filling way and embed~$P$ into~$H$
as an orthogonal node-link diagram.

\begin{figure}[tb]
	\centering
	\includegraphics[page=2]{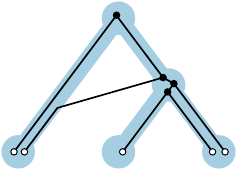}
	\caption{Representation of co-phylogenetic trees with the host tree as background shape
		and the parasite tree drawn with a node-link diagram (after Calamoneri \etal~\cite{CDMP20}).}
	\label{fig:phyTreeExamples:reconciliation}
\end{figure}

More generally, visualizations have been studied for various models in
phylogenetics, such as rooted phylogenetic trees~\cite{BBS05,Rus13,BGJO19},
in conjunction with a geographic map~\cite{Pag15,PMZ+13},
unrooted phylogenetic trees, and split networks~\cite{DH04,KH08,SNM12}.
In recent years, research has been extended from drawings of trees to
drawings of \emph{phylogenetic networks}~\cite{Hus09,HRS10,TK20,KS20,KM22},
which are more general.

All these applications share the main combinatorial objective
of finding drawings where the number of crossings between edges is minimized.
To this end, good embeddings of the trees (or networks) have to be found,
which are mostly fully defined by the order of the leaves. 
For example, Calamoneri \etal~\cite{CDMP20} investigated the problem 
of minimizing the number of crossings of the parasite tree in their drawings. 
They showed that this problem is in general NP-hard, 
though planar instances can be identified efficiently, and they
suggested two heuristics.

\paragraph{Contribution.}
Motivated by the drawing styles of Wilson \etal~\cite{WWB03}
and by the recently proposed space-filling drawing style
for reconciliation~\cite{CDMP20}
and phylogenetic networks~\cite{TK20},
we formally define tree-in-tree drawing styles for MSC trees (\cref{sec:style}).
In our base, rectangular style, we draw the species tree such that it completely fills a rectangle;
the branch widths are based on the number of leaves in the respective gene subtree. 
Additionally, population sizes can be represented, e.g., by a background color gradient.
This avoids visual overload and can be used for any population size model.
Nonetheless, based on this, we also define a style 
where the branch widths are proportional to the associated population sizes.

We then study the problem of minimizing the number of crossings between edges of the gene tree
both for the case when the embedding of the species tree is already {\em fixed}
and when it is left {\em variable}.
We show that the crossing minimization problem is NP-hard in both cases
(\cref{sec:np}). 
On the positive side, we show that crossing-free instances can be
identified in linear time (\cref{sec:planar}) and
we introduce two heuristics and an integer linear program (ILP) formulation
for the non-planar cases (\cref{sec:algo}).
We measure the performance of the heuristics on real-world instances
by comparing them to optimal solutions obtained via the ILP,
which we have tuned to solve medium-size instances in reasonable time.

Complete proofs to some of our claims and detailed descriptions 
of our algorithms can be found in \cref{app:appendix}.
Implementations of our algorithm are shared upon request.  

\section{Drawing Style}
\label{sec:style}
In this section, we define styles for tree-in-tree drawings 
of an MSC tree $\croc{S, T, \varphi}$.
A drawing is defined for particular leaf orders $\pi(S)$ 
and $\pi(T)$ of $S$ and $T$, respectively,
and we assume that they satisfy the following requirements.
(i) At least one leaf is mapped to each species.
(ii) The leaf order $\pi(T)$ is consistent with $\varphi$ and $\pi(S)$,
that is, the sets of leaves of $T$ mapped by $\varphi$ to a species $s$
are consecutive in $\pi(T)$ and 
succeed all leaves mapped to the species that precede $s$ in $\pi(S)$.
(iii) If all the leaves of a subtree $T'$ of $T$ are mapped to the
same species~$s$, then these leaves must be consecutive in~$\pi(T)$,
and $T'$ must admit a plane drawing above the leaves.
We first describe the \emph{rectangular style} where branch widths are
proportional to the number of leaves in subtrees,
and then the \emph{proportional style} where branch widths are
proportional to the population sizes.  Finally, we define the crossing
minimization problem for tree-in-tree drawings.

\paragraph{Rectangular style.}
Our drawing area is an axis-aligned rectangle $R$. 
The width of $R$ is twice the number of leaves of $T$.
We assume that the roots of~$S$ and~$T$ have out-degree~1.
We scale $h$ such that the heights of the roots equal the height of~$R$.
The given heights of vertices and nodes thus
correspond to heights in~$R$.

The species tree $S$ is drawn as follows; see~\cref{fig:style}.
For a species $s \in L(S)$, we define $n(s) = \abs{ \varphi^{-1}(s)}$,
that is, the number of leaves of $T$ mapped to $s$ by~$\varphi$.
The branches of $S$ are represented by internally disjoint rectangles whose union covers~$R$.
Of each such rectangle we only draw the left and the right border -- the \emph{delimiters}.
Their y-coordinates are defined by the heights of their start and target nodes.
The x-coordinates are defined recursively:
A branch incident to a species $s$ has width $2n(s)$
and an internal branch has width equal to the width of its two child branches; see~\cref{fig:style:base}.
Note that the branch incident to the root has a width equal to the width of $R$.

\begin{figure}[tbh]
  \centering
    \begin{subfigure}[t]{0.3 \linewidth}
		\centering
		\includegraphics[page=1]{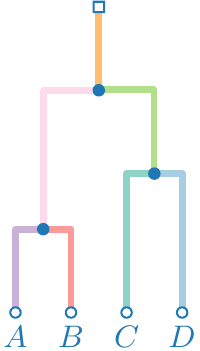}
		\caption{Species tree $S$.}
		\label{fig:style:s}
	\end{subfigure}
	$\qquad\qquad$
	\begin{subfigure}[t]{0.5 \linewidth}
		\centering
		\includegraphics[page=2]{mscStyle}
		\caption{MSC tree $\croc{S, T, \varphi}$.}
		\label{fig:style:base}
	\end{subfigure}
  \caption{In the rectangular drawing style for an MSC tree, the
    branch widths are proportional to the number of leaves in the
    contained subtree.}
  \label{fig:style}
\end{figure}

The gene tree $T$ is drawn in a classical orthogonal cladogram style:
The leaves of $T$ are evenly distributed at the base of $R$
by placing them on odd coordinates and ordered by $\pi(T)$.
Since $\pi(T)$ is consistent with $\varphi$ and $\pi(S)$,
for each species~$s$, the leaves $\varphi^{-1}(s)$ are thus placed at the baseline of the branch incident to~$s$.
Each inner vertex of $T$ is centered horizontally between its two
children and placed at its respective height.

Note that two or more vertical line segments can end up with the same
x-coordinate.  Suppose that this is the case for two vertical line
segments~$e_u$ and~$e_v$ that also overlap vertically
and that have end vertices $u$ and $v$, respectively;
see~\cref{fig:style:base} for an example.
Further suppose that~$e_u$ ends below~$e_v$.
We then shift $u$ slightly in the direction of its parent
and $v$ into the opposite direction.
Overlaps of horizontal line segments could be handled analogously,
though one would have to point out that the given heights are then misrepresented.

In this style, the population sizes are not represented by the branch
widths.  Instead, one could set the background color of each branch to
a corresponding intensity.  We advocate the rectangular tree-in-tree
style (with or without coloring) since it yields a clear
representation for MSC trees.  This is helpful for model diagnosis and
for finding incomplete lineage sorting events.

\paragraph{Proportional style.}
The proportional style conceptually follows the rectangular style,
though here the population sizes of each branch are represented by its width in the drawing;
see \cref{fig:mscExamples:continuous} for an example.
We require that the shape of $S$ is symmetric with respect to the central vertical axis.
Therefore, each branch~$e$ is represented by a sequence of trapezoids
whose widths are derived from the population sizes associated with~$e$.
Embedding $T$ into $S$ may force the non-horizontal line segments 
of~$T$ to take various different slopes, ``following'' the trapezoids.
The combinatorial properties of a drawing in the rectangular style and
a drawing in the proportional style may thus differ.

A proportional-style drawing can be computed as follows.
First, we draw $S$ bottom to top by adding one row of trapezoids
for each inner node $u$ of $S$ encountered. 
More precisely, at the height of $u$,
we calculate the width of each ``active'' branch and the total width of $S$.
We can then extend the delimiters between the branches. 
Second, the width of each species $s$ is subdivided
into~$2|\varphi^{-1}(s)|$ pieces at the baseline, such that each gene leaf
can be placed at an odd position and according to $\pi(T)$.
The rest of the gene tree can then be computed bottom-up. 
For each inner vertex $v$ of $T$, 
a sequence of line segments is drawn 
from each of its two children up to the height of $v$.
The two ends are connected with a horizontal line segment on which $v$ is placed centrally.
The slope of a non-horizontal line segment~$f$ is set 
such that~$f$ splits the top edge and the bottom edge
of the trapezoid containing~$f$ in the same ratio.

\paragraph{Crossing minimization problems.}
In the drawing styles above, by our assumptions for $S$, $T$, and $\varphi$,
no edge of~$T$ crosses a segment that represents~$S$. 
However, edges of~$T$ may cross each other.
Such crossings are determined by $\pi(S)$ and $\pi(T)$.  We do not
know~$\pi(T)$, and we consider two subproblems: in one $\pi(S)$ is
given, in the other $\pi(S)$ is not given.  Our objective is to find a
leaf order of~$\pi(T)$ and possibly of~$\pi(S)$ such that the number
of crossings among edges of~$T$ is minimized.

We define this problem formally for both drawing styles. 
In the \textsc{Variable Tree-in-Tree Drawing Crossing Minimization} (\VTT) problem, 
we are given an MSC tree $\croc{S, T, \varphi}$ and an integer $k$,
and the task is to find a tree-in-tree drawing (in rectangular or
proportional style) such that $T$ has at most $k$ crossings; 
a solution is specified by leaf orders $\pi(S)$ and $\pi(T)$.
In the \textsc{Fixed Tree-in-Tree Drawing Crossing Minimization}
(\FTT) problem, we have the same task, but we are additionally given a
leaf order $\pi(S)$; a solution is specified by a leaf order~$\pi(T)$.

\section{NP-Hardness}
\label{sec:np}
In this section, we show that \VTT and \FTT are NP-complete.  For
showing hardness, we reduce from \textsc{Max-Cut}, which is
NP-hard~\cite{GJ79}.  Recall that in an instance of \textsc{Max-Cut},
we are given a graph $G$ and a positive integer $c$.  The task is to
decide whether there exists a bipartition $\set{A, B}$ of the vertex
set $V(G)$ of~$G$ such that at least $c$ edges have one endpoint in
$A$ and one endpoint in $B$.

In the proofs below, we use the rectangular style.
Since the branch widths of the rectangles can also be seen as
population sizes, the proofs also hold for the proportional style.
We make use of the following construction 
where we replace a single leaf with a particular subtree.
Let $\ell$ be a leaf of $T$ with its parent $p$ at height $h(p)$.
Suppose that we replace $\ell$ with a full binary subtree~$T_\ell$ 
that has a specific number of leaves, say $n_\ell$ many.  (Recall that
a binary tree is \emph{full} if every vertex has either~0 or~2 children.)
Now we have two options to influence the shape of~$T_\ell$ in the
solution drawing.  In option~1, we set the height of the lowest
inner vertex of~$T_\ell$ to at least $h(p) - \varepsilon$
for some appropriately small $\varepsilon > 0$.
Now if the vertical line segment incident to $\ell$ is initially
crossed by at least one horizontal segment, then any drawing
of~$T_\ell$ will contain at least $n_\ell$ many crossings.
In this case, we call $T_\ell$ a \emph{thick expanded leaf}.
In option~2, we set the height of the root of~$T_\ell$
to~$\varepsilon'$, for some appropriately small $\varepsilon'>0$.
Then a drawing of $T_\ell$ will require $n_\ell$ horizontal space.
In this case, we call $T_\ell$ a \emph{wide expanded leaf}.

\begin{theorem} \label{clm:vtt:np}
The {\em\VTT} problem is NP-complete.
\end{theorem}
\begin{proof}
The problem is in NP since, given an MCS tree $\croc{S, T, \varphi}$,
an integer~$k$, and leaf orders $\pi(S)$ and $\pi(T)$, we can check in
polynomial time whether the resulting drawing has at most $k$ crossings.
To prove NP-hardness, we reduce from \textsc{Max-Cut} as follows.

For a \textsc{Max-Cut} instance $(G,c)$, we construct an instance
$(\croc{S, T, \varphi}, k)$ of the \VTT problem with a species tree
$S$, a gene tree $T$, a leaf mapping $\varphi$, and a positive
integer~$k$; see \cref{fig:np:v:overview}.
Let $V(G) = \set{v_1, \ldots, v_n}$ ($n \geq 3$), let $m = \abs{E(G)}$, 
and let $\set{A, B}$ be some partition of~$V(G)$. 
Let $S$ be a caterpillar tree on $2n + 1$ species 
labeled $0, 1, 1', \ldots, n, n'$ with decreasing depth, that is,
$S$ contains phylogenetic subtrees on species sets $\set{0, 1}$, $\set{0, 1, 1'}$, \ldots, $\set{0, 1, 1', \ldots, n, n'}$.
For each $i \in \set{1, \ldots, n}$, 
we add to~$T$ a \emph{vertex gadget} (described below) to enforce that
species~$i$ and~$i'$ are on opposite sites of~0.
Then species~$i$ being to the left of~0 corresponds to $v_i$ being in~$A$, 
whereas $i$ being to the right of~0 corresponds to $v_i$ being in~$B$.
Furthermore, for each edge $\set{v_i, v_j} \in E(G)$ with $i < j$,
we add to~$T$ an \emph{edge gadget} that consists of a \emph{cherry}
(i.e., a subtree on two leaves) from $i$ to $j'$ and that induces~$n^5$
crossings if and only if $i$ and $j$ are both to the left or both to
the right of~$0$.
By construction, all pairs of vertex gadgets will induce at most $n^2$ crossings,
all pairs of edge gadgets will induce at most~$n^4$ crossings, 
and all pairs of vertex and edge gadgets will induce at most~$2n^3$ crossings.
In total, these gadgets induce at most $2n^4$ crossings (using $n\ge3$).
Hence, by setting $k = (m - c)n^5 + 2n^4$, we get that a tree-in-tree drawing of $\croc{S, T, \varphi}$ with less than $k$ crossings exists 
if and only if $G$ admits a cut containing at least $c$ edges.

\begin{figure}[tb]
  \begin{subfigure}[b]{.22\textwidth}
    \centering
    \includegraphics[page=1]{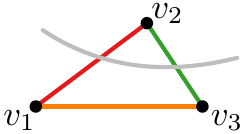}
    \caption{Input graph with a cut (gray).}
    \label{fig:np:maxCut}
  \end{subfigure}
  \hfill
  \begin{subfigure}[b]{.765\textwidth}
    \hspace*{-9.2ex}
    \includegraphics[page=2]{NPmaxCutVariable}
    \caption{Resulting rectangular tree-in-tree drawing.}
  \end{subfigure}
  \caption{Example for the reduction of a given graph to a rectangular
    tree-in-tree drawing with variable species tree embedding.  Each
    edge gadget is drawn in the respective color.}
  \label{fig:np:v:overview} 
\end{figure}

A \emph{vertex gadget} consists of two cherries; see \cref{fig:np:v:vertexGadget}.
The first cherry has one leaf each in species 0 and $i'$ and their parent $p$ gets some height $h(p)$.
The second cherry has one leaf each in species 0 and $i$ and their parent gets height $h(p) + 1$.
We replace the leaf in $i$ with a thick expanded leaf on $n^8$
many leaves.  Note that if $i$ and $i'$ are on the same side of~0,
then $i$ lies between $i'$ and~0.
Hence, in this case, the horizontal line segment through $p$ 
crosses the thick expanded leaf and causes $n^8$ crossings. 
Since $n^8 > k$, the vertex gadgets work as intended.

\begin{figure}[t]
	\centering
	\begin{subfigure}[t]{.47 \linewidth}
		\centering
		\includegraphics[page=3]{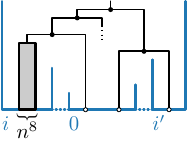}
		\caption{If $i$ and $i'$ are on different sides of $0$, the gadget induces no crossings.}
		\label{fig:np:v:vertexGadget:good}
	\end{subfigure}
	\hfill
	\begin{subfigure}[t]{.47 \linewidth}
 		\centering
		\includegraphics[page=4]{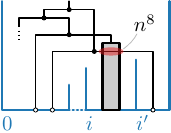}
		\caption{If $i$ and $i'$ are on the same side of $0$, the gadget induces $n^8$ many crossings.}
		\label{fig:np:v:vertexGadget:bad} 
	\end{subfigure}
  \caption{The vertex gadget for $v_i$ forces the species $i$ and $i'$ on opposite sites of species~$0$.}
  \label{fig:np:v:vertexGadget} 
\end{figure}

We set the heights of the roots of the edge gadget cherries above those of all vertex gadgets.
Furthermore, we add a thick expanded leaf on $n^5$ leaves in species 0 with the lowest inner vertex higher than any edge gadget. 
Hence, the horizontal line segment of an edge gadget crosses $n^5$ vertical segments if and only if $i$ and $j$ are both in $A$ or both in $B$. 

To tie everything together in $T$, we introduce a path from the thick
expanded leaf in~0 to the root.  To this path, going upwards, we first
connect the cherries of the vertex gadgets, whose leaves are in
$1, 1', \dots, n, n'$, in this order.  Above those, we then connect
the cherries of the edge gadgets to the path.
\end{proof}

A full proof of the following statement can be found in \cref{app:np}.

\begin{restatable}{theorem}{fttNP}\label{clm:ftt:np}
  The {\em\FTT} problem is NP-complete.
\end{restatable}
\begin{proof}[sketch]
To prove NP-hardness, we again reduce from \textsc{Max-Cut}.
For a \textsc{Max-Cut} instance $(G,c)$, we
construct an instance $(\croc{S, T, \varphi}, \pi(S), k)$ of \FTT. 
Let $V(G) = \set{v_1, \ldots, v_n}$, and
let $\set{A, B}$ be some partition of~$V(G)$.
Our construction consists of three parts and uses several different gadgets; see \cref{fig:np:f:overview}.
On the left side, we have a \emph{vertex gadget} for each vertex $v_i$.
For each edge $v_iv_j$, we have an \emph{edge gadget} that
connects the vertex gadgets of~$v_i$ and~$v_j$. The gadget has a further leaf at the far right.
We simulate $v_i$ being in either partition by having a thick expanded leaf always being either
left or right of all attached edge gadgets; otherwise it would cause too many crossings. 
Using \emph{spacer gadgets}, the leaves of edge gadgets to the far right are horizontally placed such that the root of each edge gadget 
lies exactly where we place a \emph{cut gadget}.
The cut gadget will induce $n^4$ crossings with the incoming edge of the root of each edge gadget 
only if the respective vertices are in the same partition.
While some parts of our construction induce a fixed number of crossings,
others cause in total at most $n^3$ crossings.
Hence, as in the proof of \cref{clm:vtt:np},
we can set~$k$ with respect to $c$
such that the instance admits a tree-in-tree drawing with at most $k$ crossings
if and only if $G$ admits a cut with at least $c$ edges.
\end{proof}

\begin{figure}[t]
  \centering
  \includegraphics[page=2]{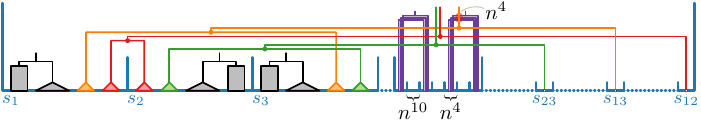}
  \caption{Sketch of the reduction of the graph from
    \cref{fig:np:maxCut} to a rectangular tree-in-tree drawing with
    fixed species tree embedding.  Each edge gadget is drawn in the
    color of the respective edge in \cref{fig:np:maxCut}.  The gadget
    for the edge~$v_1v_3$ (orange) has $n^4$ crossings more than the
    other edge gadgets; namely with the cut gadget (purple).}
  \label{fig:np:f:overview} 
\end{figure}

\section{Planar Instances}
\label{sec:planar}
In this section, we show that we can decide in linear time
whether an \FTT or \VTT instance admits a planar drawing.

\begin{theorem} \label{clm:planar:vtt}
Both when the embedding of $S$ is fixed or variable, we can decide,
in linear time, whether an MSC tree $\croc{S, T, \varphi}$ admits a
planar rectangular tree-in-tree drawing.  
If yes, such a drawing can be constructed within the same time bound.
\end{theorem}
\begin{proof}
Bertolazzi \etal~\cite{BDMT98} devised a constructive linear-time algorithm 
for upward planarity testing of a single-source (or single-sink) digraph,
that is, whether the given digraph can be drawn with each edge $uv$ 
drawn as a monotonic upward curve from $u$ to $v$.
For both the \VTT and \FTT problem, 
we can extend $T$ to a single-source digraph $\bar T$
that admits an upward planar embedding if and only if 
$\croc{S, T, \varphi}$ admits a planar tree-in-tree drawing 
(respecting any given leaf order for $S$).
We can thus apply Bertolazzi \etal's algorithm to $\bar T$.

First, suppose the embedding of $S$ is variable.
Let $L_1, L_2, \ldots, L_m$ be the subsets of $L(T)$ corresponding to the $m$ species of $S$.
For $i \in \set{1, \ldots, m}$, we merge all vertices in $L_i$ into a single vertex $v_i$.
We then connect the vertices $v_1, \ldots, v_m$ to a new vertex $t$; see \cref{fig:zeroCrossings:variable}.
We use the resulting single-source digraph as $\bar T$, 
which clearly has the desired properties.

\begin{figure}[b]
  \centering
    \begin{subfigure}[t]{0.28 \linewidth}
		\centering
		\includegraphics[page=1]{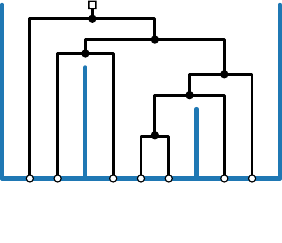}
		\caption{A planar tree-in-tree instance $\croc{S, T, \varphi}$.}
		\label{fig:zeroCrossings:input}
	\end{subfigure}
	\hfill
	\begin{subfigure}[t]{0.30 \linewidth}
		\centering
		\includegraphics[page=2]{zeroCrossings}
		\caption{Extending $T$ if the embedding of $S$ is variable.}
		\label{fig:zeroCrossings:variable}
	\end{subfigure}
	\hfill
	\begin{subfigure}[t]{0.35 \linewidth}
		\centering
		\includegraphics[page=3]{zeroCrossings}
		\caption{Extending $T$ further if the embedding of $S$ is fixed.}
		\label{fig:zeroCrossings:fixed}
	\end{subfigure}
  \caption{We can test efficiently whether a tree-in-tree instance~$\croc{S, T, \varphi}$ admits a planar drawing 
  with a single-source upward planarity test on an extended gene~tree.}
  \label{fig:zeroCrossings}
\end{figure}

Second, if the embedding of $S$ is fixed, 
we extend $\bar T$ from above further to ensure that the subsets $L_1, \ldots, L_m$ end up in correct order.
Let the species of~$S$ be $s_1, \ldots, s_m$ from left to right.
For $i \in \set{1, \ldots, m-1}$, we add a vertex $u_i$ and 
edges $v_iu_i$, $v_{i+1}u_i$, and $u_it$; see \cref{fig:zeroCrossings:fixed}.
The resulting graph is our new $\bar T$, which works again as intended.

In both cases, $\bar T$ has a linear size and can be constructed in linear time.
\end{proof}

\section{Algorithms} 
\label{sec:algo}
For non-planar instances of the \FTT and the \VTT problem,
we propose a heuristic as well as an ILP.
We describe the main ideas of the algorithms here;
more details can be found in \cref{app:heuristics}.
The ILP, which models a drawing in a straightforward fashion, is described in \cref{app:ilp}.
Overlaps of vertical segments in an ILP solution are resolved in a post-processing step. 
We focus here on the rectangular tree-in-tree style,
though the heuristics can also be set up analogously for the proportional style.
However, since the computation for the proportional style is more involved,
as alternative, one can simply use leaf orders computed for the rectangular style.

\paragraph{Heuristic for {\em\FTT}.}
Let $\croc{S, T, \varphi}$ be an MSC tree and $\pi(S)$ a leaf order for $S$.
The idea of the heuristic is to greedily sort the leaves in each species from the left and from the right towards the centre.
To this end, the algorithm (i) goes through the inner vertices in order of increasing height
and (ii) when the subtree~$T(v)$ of an inner vertex $v$ has leaves in more than one species,
then any unplaced leaves of $T(v)$ are put on a \emph{left stack} 
or a \emph{right stack} of their respective species; see \cref{fig:heuristic}. 
In doing so, we aim at a placement that minimizes the horizontal dimension of a drawing of $T(v)$.
In particular, $T(v)$ initially has unplaced leaves in at most two species.
Therefore, we place the leaves in the left species $s$ on the right stack of $s$
and the leaves in the right species $s'$ on the left stack of $s'$; see \cref{fig:heuristic:second}.
When leaves are pushed on a stack,
it is ensured that any subtree with all leaves in one species
admits a planar drawing.
This can be done in linear~time.

\begin{figure}[t]
  \centering
    \begin{subfigure}[t]{0.45 \linewidth}
		\centering
		\includegraphics[page=1]{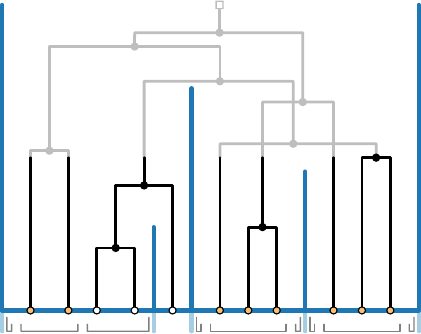}
		\caption{Configuration after four vertices.}
		\label{fig:heuristic:first}
	\end{subfigure}
	\hfill
	\begin{subfigure}[t]{0.45\linewidth}
		\centering
		\includegraphics[page=2]{heuristic}
		\caption{Configuration after six vertices.}
		\label{fig:heuristic:second}
	\end{subfigure}
  \caption{The heuristic sorts the leaves in each species from the sides towards the centre
  by using a left stack and a right stack for each species (plus a central bucket of unplaced (orange) leaves),
  here on the example from \cref{fig:mscExamples,fig:style}.}
  \label{fig:heuristic}
\end{figure}

\paragraph{Heuristic for {\em\VTT}.}
We extend the heuristic for \FTT to also compute a leaf order for $S$ as follows.
The main idea is to set the rotation of inner nodes of $S$ 
such that subtrees of $T$ horizontally span over few species.
Therefore, when we handle an inner vertex $v$ with children $x$ and $y$
and we try to move the roots of~$T(x)$ and $T(y)$ close together.
Suppose $x$ lies in the branch ending at node $x'$ of $S$.
Let $S(x')$ be the minimal phylogenetic subtree of $S$ on all species
that contain a leaf of $T(x)$; define $S(y')$ analogously.
If $S(x')$ and $S(y')$ are disjoint,
then we set the rotation of each unfixed vertex 
on the path from the root of $S(x)$ to the root of $S(y)$
such that the species of $S(x)$ and $S(y)$ get as close together as possible; see \cref{fig:heuristic:vtt:short}.
Only then is $v$ processed with the \FTT heuristic. 
There are a few other cases to consider, which can be handled along the same line
(see \cref{app:heuristics} for details).
Overall, handling an inner vertex of $T$ can be done in linear time
and so the overall running time is quadratic.

\begin{figure}[t]
  \centering
  \includegraphics[page=3]{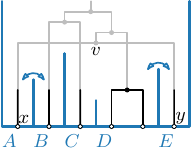}
  \caption{The heuristic for the \VTT problem rotates inner nodes of $S$
  to bring the leaves of the currently handled gene subtree closer together.
  Here, for the second inner vertex $v$ of $T$, two nodes would be rotated to bring the species $A$ and~$E$~together.}
  \label{fig:heuristic:vtt:short}
\end{figure}

Note if an instance admits a planar solution, then the heuristics find one.
That is, because any rotation of a node of $S$
or an assignment to stacks keeps the leaves of a subtree of $T$ consecutively whenever possible.

\paragraph{Experimental evaluation.}
We tested the heuristic and the ILP on three different real world data sets
Gopher ($S$ on 8 species, $T$ on 26 gene taxa, 1083 instances, 
i.e., different topologies and heights for pairs of $S$ and $T$)~\cite{DataGopher},
Barrow (21 species, 88 gene taxa, 312 instances)~\cite{DataBarrow},
and Hamilton (36 species, 83 gene taxa, 99 instances)~\cite{DataHamilton}.
On a laptop with 4 cores, 8 GB of RAM, Ubuntu 20.04, and CPLEX 12.10 
we tested each heuristic and the ILP on each instance once with the default (start) embedding of $S$ from the input file
and~10 times with a random (start) embedding for $S$. 
A proper experimental evaluation is out of scope for this paper, but we observed the following:
\begin{itemize}
  \item The \VTT heuristic got a better result than the \FTT heuristic for 60--75\% of the instances,
  the same result for 6--27\%, and a worse result for 13--20\%. 
  For the Barrow instances, they improved the average number of 24.5 crossings of the default embeddings 
  to 10.3 (\FTT) and 7.2 (\VTT) or even to 6.6 and~5.7 using random starting embeddings of $S$.
  \item Concerning \FTT, the optimal solutions found by the ILP show that the \FTT heuristic 
  also found the optimal solution for about 50--55\% of the instances;
  e.g., for the Barrow instances, the \FTT heuristic had on average only 1.3 crossings more than the optimal.
  Concerning \VTT, the heuristics also got within zero to few crossings to the best ILP solution for Gopher instances.
  \item Both heuristics are sensitive to the initial embedding of $S$ 
  as the lowest number of crossings was achieved with a random start embedding for 70--90\% and for 44--75\% of the instances
  for \FTT and for \VTT, respectively. The results between start embeddings vary more for \VTT than for \FTT.  
  \item The \FTT and the \VTT heuristic run in a fraction of a second per instance, 
  while the ILP for the \FTT problem takes about 1--4s for most instances.
  The ILP for the \VTT problem only found solutions for the Gopher instance within reasonable time for some instances.
\end{itemize}
Since the heuristics are so fast, our recommendation is to run both heuristics for several different start embeddings of $S$
and then take the best found solution.

To the best of our knowledge, this is the first software to visualize MSC trees for the continuous linear model
and so we hope that this will help researchers in the emerging field of MSC to visualize their results.

\section*{Acknowledgments}
We thank the reviewers for their comments and 
J.~Douglas for providing us with the test data
and his helpful explanations concerning MSC.

\pdfbookmark[1]{References}{References} 
\bibliographystyle{abbrvurl}
\bibliography{sources}

\clearpage
\pdfbookmark[0]{Appendix}{toc.appendix}
\section*{Appendix}
\appendix
\label{app:appendix}

\section{Full Proof of Theorem 2}
\label{app:np}

\fttNP*
\begin{proof}
As with the \VTT problem, it is easy to see that the \FTT problem is in NP.
Given an instance $\croc{S, T, \varphi}, \pi(S), k$ as well as a leaf order $\pi(T)$, 
we can check in polynomial time whether this yields a drawing with at most $k$ crossings.
To prove NP-hardness, we use a reduction from \textsc{Max-Cut}.

For a \textsc{Max-Cut} instance $G, c$, we
construct an instance $\croc{S, T, \varphi}, \pi(S), k$ of the \FTT problem 
by devising a species tree $S$ with leaf order $\pi(S)$, a gene tree~$T$, 
a leaf mapping $\varphi$, and a positive integer $k$.
Let $V(G) = \set{v_1, \ldots, v_n}$ and 
let~$\set{A, B}$ be some partition of $V(G)$.
Our construction consists of three parts and uses several different gadgets, 
which are described in detail below.
\Cref{fig:np:f:overview:app} shows our constructions for the \textsc{Max-Cut} instance from \cref{fig:np:maxCut}.
On the left side, we have a \emph{vertex gadget} for each vertex $v_i$
where we simulate $v_i$ being in either partition.
To these vertex gadgets, we connect an \emph{edge gadget} for each edge
that also has a leaf on the right side.
Using \emph{spacer gadgets}, the third leaf of an edge gadget is horizontally placed such that the root of the edge gadget 
lies in the centre, where we place a \emph{cut gadget}.
The cut gadget will induce $n^4$ crossings with the incoming edge of the root of each edge gadget 
only if the respective vertices are in the same partition.
While some parts of our construction induce a fixed number of crossings,
others cause in total at most $n^3$ crossings.
Hence, as in the proof of \cref{clm:vtt:np},
we can set $k$ with respect to $c$
such that the instance admits a drawing with at most $k$ crossings
if and only if $G$ admits a cut with at least~$c$~edges.

\begin{figure}[b]
  \centering
  \includegraphics[page=2]{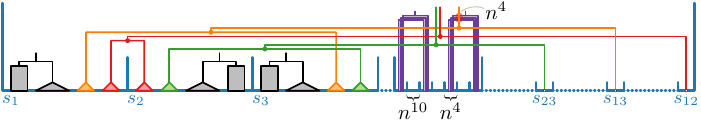}
  \caption{Sketch of the reduction of the graph from \cref{fig:np:maxCut} 
  		to a rectangular tree-in-tree drawing with fixed species tree embedding.
		Each edge gadget is drawn in the respective color; 
		the one for $v_1v_3$ has $n^4$ more crossings with the {\pu cut gadget}.}
  \label{fig:np:f:overview:app} 
\end{figure}

A \emph{spacer gadget} consists of one species containing a wide expanded leaf of desired width.
The gene tree parts of all spacer gadgets are tied together 
with a tree that has all internal vertices above all other gadgets,
which in turn is connected to the rest of $T$ close to the root.
Hence, the spacer gadgets induce a fixed number of crossings.     

A \emph{vertex gadget} for a vertex $v_i$
consists of seven species, namely, $s_i^j$ with $j \in \set{-3, -2, -1, 0, 1, 2, 3}$; see \cref{fig:np:f:vertexGadget}.
For each edge incident to $v_i$,
there is one wide expanded leaf of width $n^7$ of the respective edge gadget in $s_i^0$.
Moreover,~$s_i^0$ contains a \emph{partitioner} tree
that consists of a thick expanded leaf of width~$n^7$
and a wide expanded leaf of width~$n^{10}$.
There are two cherries with leaves in~$s_i^{-3}$ and~$s_i^1$
as well as in~$s_i^{-1}$ and~$s_i^3$. 
The cherries are connected with a vertex~$p_i$.
By setting the heights appropriately and 
with one spacer gadget each in~$s_i^{-2}$ and in~$s_i^2$,
we enforce that the partitioner cuts always through the horizontal line segments of both cherries
(causing always~$2n^7$ crossings).
Furthermore, only if the partitioner is all the way to the left or all the way to the right in~$s_i^0$,
then its thick expanded leaf does not cross with the horizontal line segment through~$p_i$; see again \cref{fig:np:f:vertexGadget}.
Otherwise, there are another~$n^7$ crossings which is always more than~$c$ 
and would thus obstruct any solution.
Hence, the partitioner being at the left or at the right corresponds to~$v_i$ being in~$A$ or~$B$, respectively.
The wide expanded leaf of the partitioner, which always lies next to the thick expanded leaf 
and below the horizontal line segment of~$p_i$,
has the effect that the horizontal shift for the attached tree gadgets is substantial enough
between the two configurations.

\begin{figure}[t]
  \centering
      \begin{subfigure}[t]{0.80 \linewidth}
		\centering
		\includegraphics[page=3]{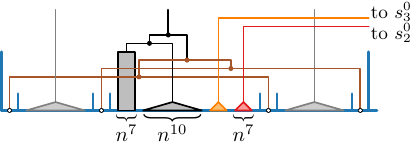}
		\caption{A partitioner to the left corresponds to $v_0$ being in $B$ of the partition $\set{A, B}$.}
		\label{fig:np:f:vertexGadget:A}
	\end{subfigure}
	\begin{subfigure}[t]{0.80 \linewidth}
		\centering
		\includegraphics[page=5]{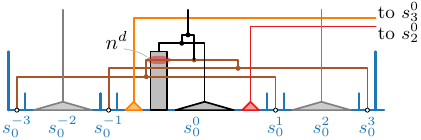}
		\caption{A central partitioner would induce additional $n^7$ crossings.}
		\label{fig:np:f:vertexGadget:bad}
	\end{subfigure}
  \caption{The vertex gadget for $v_0$ with its seven species has two edge gadgets attached.}
  \label{fig:np:f:vertexGadget} 
\end{figure}

We assume that the edges are lexicographically ordered based on the indices of their vertices. 
The \emph{edge gadget} for an edge~$v_iv_j$
contains a cherry on the expanded wide leaves in~$s_i^0$ and~$s_j^0$
and has the root~$p_{ij}$ that connects the cherry with another leaf in a species~$s_{ij}$.
The height of~$p_{ij}$ is set below some value~$y^\star$.
With spacer gadgets between these species (such as~$s_{ij}$) on the right,
we can enforce that~$p_{ij}$ lies in the central gap of the cut gadget
only if (i) the expanded leaves of edge gadgets in~$s_i^0$ and~$s_j^0$ are lexicographic ordered from left to right
and (ii)~$v_i$ and~$v_j$ are in different partitions.
However, if the order is different or~$v_i$ and~$v_j$ are in the same partition,
then~$p_{ij}$ lies inside a tower of the cut gadget.
The gene tree parts of the edge gadgets merge into a tree above the cut~gadget.

The \emph{cut gadget} consists of two \emph{towers} 
with a spacer gadget of~$\Oh(1)$ width between them; see again \cref{fig:np:f:overview}.
Each tower spans~$2n^4+1$ species, where the central one contains a spacer gadget of width~$n^{10}$.
It further contains~$n^4$ cherries where the two leaves 
of each cherry are the only leaf in their respective species
and that are~$n^4+1$ species apart.
The root of each such cherry gets a height between~$y^\star$ and~$y^\star + \varepsilon$
for some appropriate small~$\varepsilon$.
The cherries of each tower and then the two towers are connected into a single tree, 
which gets connected to the rest of~$T$ somewhere close to the root.
Now, if~$p_{ij}$ lies above the central species of the cut gadget,
then the vertical segment incident to~$p_{ij}$ does not cut through the tower.
If however~$v_i$ and~$v_j$ are, say, both in~$A$, then~$p_{ij}$ is placed
further to the left (the partitioner causes a shift of~$(n^{10} + n^7)/4$).
In particular, even if the edge gadgets in~$s_i^0$ and~$s_j^0$ have a different order,
$p_{ij}$ is placed inside the tower.
Hence, in this case, the vertical segment incident to~$p_{ij}$ cuts through the~$n^4$ cherries of the left tower.
In other words,~$v_iv_j$ not being in the cut corresponds to an additional~$n^4$ crossings.
\end{proof}

\section{Heuristics}
\label{app:heuristics}
In this section, we describe the heuristics that try to minimize the number of crossings
in a tree-in-tree drawing of an MSC tree $\croc{S, T, \varphi}$ in more detail.
We consider first the rectangular style for both the case 
when a leaf order $\pi(S)$ is given (\FTT problem) and when it is not (\VTT problem).
Afterwards we explain what changes have to be made to the heuristics
when using them for the proportional style.

\subsection{Heuristic for the \FTT Problem}
\label{app:sec:ftt}
Let $\croc{S, T, \varphi}$ be an MSC tree and let $\pi(S)$ be a leaf order of $S$.
Based on $\pi(S)$ the drawing of $S$ can already be computed.
The heuristic then draws $T$ within~$S$ in the rectangular tree-in-tree style.
Recall that to this end, the heuristic 
(i) goes through the inner vertices in order of increasing height
and (ii) when the subtree~$T(v)$ of an inner vertex~$v$ has leaves in more than one species,
then any unplaced leaves are put on a left stack or a right stack of their respective species.
More precisely, for each species~$s$, 
the leaves in~$s$ are sorted from the left towards the centre with a \emph{left stack}
and from the right towards the centre with a \emph{right stack}.
Initially no leaves of~$T$ are on any left or right stack. 
We thus say that a leaf is \emph{unfixed} and 
that it lies in the \emph{central bucket} of its respective species; see \cref{fig:heuristic}.
On the other hand, a leaf is called \emph{fixed} if it is put on the left or right stack of its respective species.
An inner vertex~$v$ is called \emph{fixed} if all leaves in~$T(v)$ are fixed.
In this case, the full drawing of~$T(v)$ has been computed.
At the end all leaves are fixed and the drawing of~$T$ is completed.
The order of the leaves in a species has emerged as the concatenation of its two stacks. 
Furthermore, it is ensured that any subtree~$T'$ fully contained within a species
the leaves of~$T'$ are sorted such that~$T'$ is drawn planar. 

During the executing of the heuristic,
we maintain for each vertex whether it is fixed or unfixed.
Note that an inner vertex~$v'$ can only be unfixed
if all leaves in~$T(v')$ lie in the same species~$s$.
In this case, we also store the species~$s$ for~$v'$.

We now describe the different cases when handling an inner vertex~$v$.
Let~$x$ and~$y$ be the two children of~$v$. 
\begin{enumerate}[leftmargin=*,label=Case F\arabic*:,]
\item Suppose that all leaves of~$T(v)$ lie in the same species~$s$. 
Then~$x$,~$y$, and~$v$ remain unfixed and the leaves of~$T(v)$ remain in the central bucket of~$s$. 
Using the information stored for~$x$ and~$y$, this case can be handled in~$\Oh(1)$~time.
\item Suppose that~$x$ and~$y$ are unfixed and 
that the leaves of~$T(x)$ and~$T(y)$ lie in distinct species~$s_x$ and~$s_y$ respectively.
Without loss of generality, assume that~$s_x$ lies left of~$s_y$.
Then the leaves in~$T(x)$ are put on the right stack of~$s_x$ 
in an order such that~$T(x)$ can be drawn planar. 
The drawing of~$T(x)$ is then computed; see again \cref{fig:heuristic:ftt}.
$T(y)$ is handled analogously using the left stack of~$s_y$.
This takes~$\Oh(\abs{L(T(v))})$ time.
\item Suppose that~$x$ is unfixed and~$y$ is fixed.
Let~$s_x$ be the species containing the leaves of~$T(x)$.
If~$y$ lies to the left (right) or above the left (resp. right) stack of~$s_x$, 
then we place the leaves of~$T(x)$ on the left (resp. right) stack of~$s_x$.
If~$y$ lies between the left and right stack of~$s_x$,
then we place the leaves of~$T(x)$ on, say, the left stack of~$s_x$.
As before, the leaves are added to a stack such that a planar drawing of~$T(x)$ can be computed.
This case takes~$\Oh(L(T(x)))$~time.
\item If both~$x$ and~$y$ are already fixed, then so is~$v$.
The drawing of~$T(v)$ is completed by drawing~$v$ and
the edges~$vx$ and~$vy$ in~$\Oh(1)$ time.
\end{enumerate}
Note that the number of fixed and drawn vertices correlates linearly 
with the running time. Hence, the heuristic runs in~$\Oh(n)$ time.

\begin{figure}[tbh]
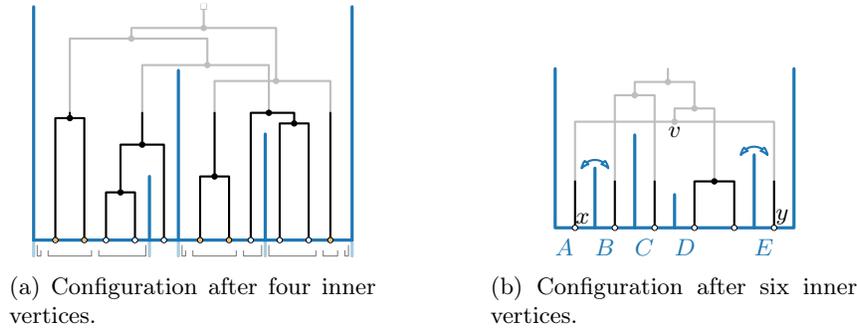

  \centering
    \begin{subfigure}[t]{0.40 \linewidth}
		\centering
		\includegraphics[page=2]{heuristic}
		\caption{Configuration after four inner vertices.}
		\label{fig:heuristic:ftt:first}
	\end{subfigure}
	$\qquad\qquad$
	\begin{subfigure}[t]{0.40\linewidth}
		\centering
		\includegraphics[page=3]{heuristic}
		\caption{Configuration after six inner vertices.}
		\label{fig:heuristic:ftt:second}
	\end{subfigure}
  \caption{The heuristic sorts the leaves in each species from the sides towards the centre
  by using a left stack and a right stack for each species (plus a central bucket of unplaced (orange) leaves),
  here on the example from \cref{fig:mscExamples,fig:style}.}
  \label{fig:heuristic:ftt}
\end{figure}

\subsection{Heuristic for the \VTT Problem}
For the \VTT problem, our heuristic computes for a given an MSC tree~$\croc{S, T, \varphi}$
a leaf order~$\pi(S)$ given by the rotation of each inner node of~$S$
alongside a leaf order~$\pi(T)$.
In order to use the heuristic for the \FTT problem to compute~$\pi(T)$,
when handling an inner vertex~$v$ of~$T$,
the order of the species that contain leaves of~$T(v)$ is set first.
We say a node of~$S$ is \emph{set} if its rotation has been set;
otherwise we say it is \emph{unset}.
Here the idea to minimize crossings is to set the rotations of the inner nodes of~$S$
such that subtrees of~$T$ horizontally span over few species.

Suppose that we handle an inner vertex~$v$ with children~$x$ and~$y$.
We would then like~$x$ and~$y$ to be close together.
Suppose $x$ lies in the branch ending at node $x'$ of $S$
and $y$ lies in the branch ending at node $y'$ of $S$.
Let~$S(x')$ be the minimal rooted phylogenetic subtree of~$S$ on all species
that contain a leaf of~$T(x)$; define~$S(y')$ analogously.
Note that $S(x')$ and $S(y')$ are rooted phylogenetic subtrees of $S$ and hence contain no nodes of degree two.
Furthermore, they may contain species that do not contain a leaf of $T(x)$ or $T(y)$, respectively.
We distinguish the following cases.
\begin{enumerate}[leftmargin=*,label=Case V\arabic*:,]
  \item Suppose that~$S(x')$ and~$S(y')$ are disjoint.
  Let~$w$ be the lowest common ancestor of~$x'$ and~$y'$ in~$S$.
  Without loss of generality, assume that~$x'$ lies in the left subtree of~$w$.
  Then set the rotation of each unset node on the path from~$x'$ to~$w$ (excluding~$w$)
  such that~$x'$ lies in its right subtree.
  Analogously, set the rotation of each unset node on the path from~$y'$ to~$w$ (excluding~$w$)
  such that~$y'$ lies in its left subtree.
  An example is shown in \cref{fig:heuristic:vtt:first}.
  \item Suppose that~$S(y')$ is a subtree of~$S(x')$
  and that no leaf of~$T(x)$ lies in a species of~$S(y)$.
  On the path from~$y'$ to~$x'$,
  let~$w$ be the first node with a species in its subtree that contains a leaf of~$T(x)$.
  Without loss of generality, assume that~$y'$ lies in the right subtree of~$w$.
  Then set the rotation of each unset node on the path from~$y'$ to~$w$ (excluding~$w$)
  such that~$y'$ lies in its left subtree; see \cref{fig:heuristic:vtt:second}.
  The case where~$S(x')$ and~$S(y')$ have reversed roles is handled analogously.
  \item Suppose that neither of the two previous cases applies. 
  Note that then either~$y'$ lies on the path from a species containing a leaf of~$T(x)$ to~$x'$
  or this occurs with~$x'$ and~$y'$ in reversed roles.
  Since all nodes on such a path are already set,
  the species of~$S(x')$ and~$S(y')$ that contain leaves of~$T(x)$ and~$T(y)$ cannot
  be moved closer together; this might not even be clearly defined. 
  Hence, in this case, no node of~$S$ gets set.
\end{enumerate}
Note that after the respective case has been applied,
the relative order of the species containing leaves of~$T(v)$ is fixed.
The heuristic can thus proceed with~$T(v)$ as described in \cref{app:sec:ftt}.
However, the whole subtree of~$S$ containing~$T(v)$ can later change its horizontal position and be horizontally mirrored.
Therefore, the x-coordinate of a handled vertex~$v$ is stored relative to its parent
or relative to~$S(v')$ if its parent has not been handled yet.

To determine which case applies and then to execute its respective procedure,
the heuristic has to determine~$S(x')$ and~$S(y')$ as well as walk along paths in~$S$.
This requires a constant number of traversals of~$S$ and~$T$ and
therefore for each inner vertex of~$T$ can be handled in~$\Oh(n)$ time. 
This yields an overall running time in~$\Oh(n^2)$.

\begin{figure}[tbh]
  \centering
    \begin{subfigure}[t]{0.3 \linewidth}
		\centering
 		\includegraphics[page=3]{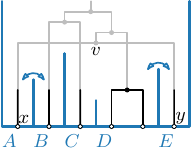}
		\caption{For the second inner vertex of $T$,
		Case 1 applies since the subtrees $S(x)$ and $S(y)$ are disjoint; 
		two nodes of $S$ are rotated.}
		\label{fig:heuristic:vtt:first}
	\end{subfigure}
	\hfill
	\begin{subfigure}[t]{0.3 \linewidth}
		\centering
 		\includegraphics[page=4]{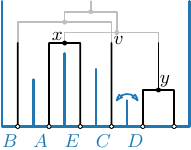}
		\caption{For the third inner vertex of $T$,
		Case 2 applies since the subtree $S(y)$ is a proper subtree of $S(y)$ 
		and no leaf of $T(x)$ lies in species $D$;
		one node of $S$ is rotated.}
		\label{fig:heuristic:vtt:second}
	\end{subfigure}
	\hfill
	\begin{subfigure}[t]{0.3 \linewidth}
		\centering
		\includegraphics[page=5]{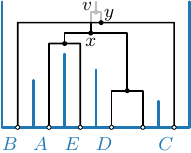}
		\caption{For the root of $T$, Case 3 applies. The embedding of $S$ is also already fixed.}
		\label{fig:heuristic:vtt:third}
	\end{subfigure}
  \caption{The heuristic for the VTT problem rotates inner nodes of the species tree 
  to bring the leaves of the currently handled gene subtree closer together. The stacks and central bucket are not shown.}
  \label{fig:heuristic:vtt}
\end{figure}

\subsection{Heuristics for the Proportional Style}
The \FTT heuristic works the same for the proportional style as it does for the rectangular style.
The only noteworthy difference is that first computing the drawing of $S$ 
and then computing the drawings of subtrees of $T$ is a bit more~involved.

Using the \VTT heuristic for the proportional style requires that 
the computed drawings of subtrees of $T$ are stored differently 
if one wants to have a quadratic running time.
This is the case because when a node $u$ of $S$ gets rotated,
the drawings of a subtree of $T$ in the subtree $S(u)$
before and after the rotation might have different slopes for every edge.
Therefore, instead of storing the x-coordinate or x-offset of a vertex $v$, 
we can store the position of $v$ relative to the total width of the branch $v$ lies in.
For example, if $v$ initially lies at 20\% of the width  from the left delimiter, 
then it will lie 20\% of the width from the right delimiter after the rotation.
Once the heuristic has reached the root of $T$,
the drawing of $T$ and in particular of the edges of $T$ can be completed
using these relative positions.

\section{ILP Formulation}
\label{app:ilp}

In this section, we describe an ILP formulation for the \VTT problem.
We are thus given an MSC tree $\croc{S, T, \varphi}$ and want
to compute leaf orders $\pi(S)$ and~$\pi(T)$ for $S$ and $T$, respectively,
such that a rectangular tree-in-tree drawing of~$\croc{S, T, \varphi}$
with $\pi(S)$ and~$\pi(T)$ has the minimum number of crossing among all
rectangular tree-in-tree drawings of $\croc{S, T, \varphi}$.
We can obtain a formulation for the \FTT problem by turning certain variables
into constants according to the given input as described at the end of this section.

We first describe the variables and constraints to model the input $S$, $T$, and~$\varphi$.
Next, we show how we model $\pi(S)$, $\pi(T)$ and coordinates to describe a rectangular tree-in-tree drawing.
We can then explain how crossings are recognized and how the number of crossings can thus be minimized.  

\paragraph{Model of input.}  
Recall that $T$ and $S$ have vertex sets $V(S)$ and $V(T)$ and leaf sets $L(S)$ and $L(T)$, respectively.
We let $\rho_S$ and $\rho_T$ denote the roots of $S$ and $T$, respectively.
The input is then described as follows.
First, the constants that model $S$:
\begin{align*}
        n_S   &= \abs{L(S)} \\
	\alpha(s) &\in V(S) \text{ is the first child of $s$ } && \text{for } s \in V(S) \setminus L(S)\\
	\beta(s) &\in V(S) \text{ is the second child of $s$ } && \text{for } s \in V(S) \setminus L(S)\\
	\gamma(s) &\in V(S) \text{ is the parent of $s$ } && \text{for } s \in V(S) \setminus \set{\rho_S}\\
\end{align*}
\noindent
Next, $T$ is modelled analogously. However, here we also need the heights (y-coordinates) of the vertices
as well as the subset of leaves of a subtree $T(v)$, that is, the \emph{clade} $L(v)$ of $v$.
\begin{align*}
        n_T   &= \abs{L(T)} \\
	\alpha(v) &\in V(T) \text{ is the first child of $v$ } && \text{for } v \in V(T) \setminus L(T)\\
	\beta(v) &\in V(T) \text{ is the second child of $v$ } && \text{for } v \in V(T) \setminus L(T)\\
	\gamma(v) &\in V(T) \text{ is the parent of $v$ } && \text{for } v \in V(T) \setminus \set{\rho_T}\\
	y(v) &\in \mathbb{Q}^+ \text{ is the y-coordinate of $v$ } && \text{for } v \in V(T)\\
	L(v) &\subseteq L(T) \text{ is the clade of $v$ } && \text{for } v \in V(T)\\
\end{align*}
\noindent
To model $\varphi$, we use the following constants:
\begin{align*}
	\varphi(v) &\in V(S) \text{ is the root of $S(v')$ (details below)} && \text{for } v \in V(T)\\
	\varphi^{-1}(s) &\subset V(T) \text{ is the set of leaves mapped to $s$ } && \text{for } s \in V(S)\\
\end{align*}
\noindent
Note that we extended $\varphi$ to inner vertices of $T$.
More precisely, for an inner vertex~$v$ of $T$ that lies in the branching ending at node $v'$ of $S$,
we have that $\varphi(v)$ is the root of $S(v')$
with, as we recall from above, $S(v')$ defined as the minimal phylogenetic subtree of $S$
that contains the clade of $v$.
If the clade of $v$ is completely mapped to a species, (i.e. $\varphi(v) \in L(S)$), 
we will enforce that the clade $L(v)$ forms an interval among the leaves in $L(T)$.

Note that a crossing can only occur between a horizontal line segment $f_u$ and a vertical line segment $g_v$.
Suppose $f_u$ goes through a vertex $u$ of $T$.
The y-coordinate of $f_u$ is thus given by the y-coordinate of $u$
and the x-coordinates of $f_u$ are given by the x-coordinates of the two children $\alpha(u)$ and $\beta(u)$ of $u$.
Suppose $g_v$ ends at a vertex $v$.
The x-coordinate of $g_v$ is thus given by the x-coordinate of $v$
and the y-coordinates are given by the y-coordinates of $v$ and the parent $\gamma(v)$ of $v$.
To decide whether $f_u$ and $g_v$ cross, we have to determine 
whether they overlap both vertically and horizontally. 
Since a vertical overlap is fully determined by the input, we can precompute this
and pass it to ILP as $a(u,v)$:
\begin{align*}
	a(u,v) &\in \set{0, 1} \text{ whether $f_u$ and $g_v$ overlap vertically} && \text{for } u\in V(T) \setminus L(T)\\
		&								&& \text{ and }v \in V(T) \setminus \set{\rho_T}
\end{align*}

\paragraph{Model of drawing.}
The main variables to model a drawing are the variables for the x-coordinates of the vertices.
In fact, only the variables for x-coordinates of the leaves are actually decision variables,
while all other variables will be auxiliary variables required to check feasibility
or to count crossings. In the following description of the variables,
$f_u$ is the horizontal line segment passing through $u$ and $g_v$ is a vertical line segment ending at $v$ 
as defined above. 
\begin{align*}
	\intertext{\textnormal The x-coordinates of the vertices of $T$:}
	x_v &\in \set{1, \ldots, n_T} && \text{for } v \in L(T)\\
	x_v &\in \mathbb{Q} && \text{for } v \in V(T) \setminus L(T)\\
	\intertext{\textnormal The intended meaning of $\bar{l}_{u, v} = 0$ ($\bar{r}_{u, v} = 0$) is that $g_v$ is left (resp. right) of $f_u$:}	
	\bar{l}_{u, v} &\in \set{0, 1} && \text{for } u \in V(T) \setminus L(T), v \in V(T) \setminus \set{\rho_T}\\
	\bar{r}_{u, v} &\in \set{0, 1} && \text{for } u \in V(T) \setminus L(T), v \in V(T) \setminus \set{\rho_T}\\
	\intertext{\textnormal The intended meaning of $z_{u, v} = 1$ is that $f_u$ horizontally overlaps with $g_v$:}
	z_{u, v} &\in \set{0, 1} && \text{for } u \in V(T) \setminus L(T), v \in V(T) \setminus \set{\rho_T}\\
	\intertext{\textnormal The intended meaning of $I^l_v = 42$ ($J^l_s = 42$) is that the
	leftmost leaf (resp. species) in the clade of $v$ (resp. $s$) has position $42$ in the permutation of $L(T)$ (resp. $L(S)$).}
	I^l_v &\in \set{1, \ldots, n_T} && \text{for } v \in V(T)\\
	I^r_v &\in \set{1, \ldots, n_T} && \text{for } v \in V(T)\\
	J^l_s &\in \set{1, \ldots, n_T} && \text{for } s \in V(S)\\
	J^r_s &\in \set{1, \ldots, n_T} && \text{for } s \in V(S)
\end{align*}
Note that we use the notation $x_v$ for the x-coordinate of $v$ but $y(v)$ for the y-coordinate of $v$,
since $x_v$ is a variable whereas $y(v)$ is a given constant.

For a permutation to represent a feasible solution we need to check
whether the leaves mapped to each species form a consecutive partial permutation.
This can be achieved by propagating interval bounds from the leaves towards the root of $S$. 
The equivalent property must be checked for inner vertices of $T$ 
whose clade is contained within a single species.

Concerning the overlap of $f_u$ and $g_v$, note
that they overlap horizontally only if $x_v \in [x_{\alpha(u)}..x_{\beta(u)}]$.
This can be check with the equation $((x_{\beta(u)} - x_v) \cdot (x_v - x_{\alpha(u)})) \ge 0$.
In order to linearize these constraints 
we have an intermediate step and the auxiliary variables $\bar{l}$ and $\bar{r}$.
These check whether $x_v$ is completely left or right of the interval, respectively.
In accordance with constraints below, the solver may then set $\bar{l}_{u, v}$ or $\bar{r}_{u, v}$ to 0, respectively,
in order to minimize the objective function.
If on the other hand neither is true, $x_v$ lies in $[x_{\alpha(u)}..x_{\beta(u)}]$ and a horizontal overlap occurs.

We use the following set of constraints to enforce the desired behavior for all the variables. 
\begin{align*}
	\intertext{\textnormal Place leaves into the interval of the species they belong to:}
	x_v &\geq J^l_{\varphi(v)} && \text{for } v \in L(T)\\
	x_v &\leq J^r_{\varphi(v)} && \text{for } v \in L(T)\\
	\intertext{\textnormal Distinct leaves get distinct x-coordinates:}
	x_u &\neq x_v && \text{for } u,v\in L(T), u \neq v\\
	\intertext{\textnormal Compute x-coordinates of inner vertices of $T$:}
	x_v &= (x_{\alpha(v)} + x_{\beta(v)}) / 2 && \text{for } v \in V(T) \setminus L(T)\\
	\intertext{\textnormal Examine horizontal precendences:}
	\bar{l}_{u, v}&\geq (x_v - x_{\alpha(u)})/n_T && \text{for } u \in V(T) \setminus L(T), v \in V(T) \setminus \set{\rho_T}\\
	\bar{l}_{u, v}&\geq (x_v - x_{\beta(u)})/n_T && \text{for } u \in V(T) \setminus L(T), v \in V(T) \setminus \set{\rho_T}\\
	\bar{r}_{u, v}&\geq (x_{\alpha(u)} - x_v)/n_T && \text{for } u \in V(T) \setminus L(T), v \in V(T) \setminus \set{\rho_T}\\
	\bar{r}_{u, v}&\geq (x_{\beta(u)} - x_v)/n_T && \text{for } u \in V(T) \setminus L(T), v \in V(T) \setminus \set{\rho_T}\\
	\intertext{\textnormal Calculate whether $g_v$ horizontally overlaps with $f_u$:}
	z_{u, v} &\geq \bar{l}_{u, v} + \bar{r}_{u, v} -1 && \text{for } u\in V(T) \setminus L(T), v \in V(T) \setminus \set{\rho_T}\\
	\intertext{\textnormal Propagate interval limits bottom-up through $T$:}
	I^l_v & \leq I^l_{\alpha(v)} && \text{for } v \in V(T) \setminus L(T)\\
	I^l_v & \leq I^l_{\beta(v)} && \text{for } v \in V(T) \setminus L(T)\\
	I^r_v & \geq I^r_{\alpha(v)} && \text{for } v \in V(T) \setminus L(T)\\
	I^r_v & \geq I^r_{\beta(v)} && \text{for } v \in V(T) \setminus L(T)\\
	\intertext{\textnormal Propagate interval limits bottom-up through $S$:}
	J^l_s & \leq J^l_{\alpha(s)} && \text{for } s \in V(S) \setminus L(S)\\
	J^l_s & \leq J^l_{\beta(s)} && \text{for } s \in V(S) \setminus L(S)\\
	J^r_s & \geq J^r_{\alpha(s)} && \text{for } s \in V(S) \setminus L(S)\\
	J^r_s & \geq J^r_{\beta(s)} && \text{for } s \in V(S) \setminus L(S)\\
	\intertext{\textnormal Check interval limits where appropriate:}
	I^l_v &= x_v && \text{for } v \in L(T)\\
	I^r_v &= x_v && \text{for } v \in L(T)\\
	I^r_v - I^l_v + 1 &= \abs{L(v)} && \text{for } v \in V(T) \setminus L(T), \varphi(v) \in L(S)\\
	J^r_s - J^l_s + 1 &= \abs{\varphi^{-1}(s)} && \text{for } s \in V(S)
\end{align*}

\paragraph*{Counting crossings.}
Since a crossing happens if and only if the two line segments overlap both
horizontally and vertically, the objective function simply needs to count
how often these two properties coincide. Note that $a(u,v)$ is part of the
input and hence the objective function is linear:

\[\text{minimize }\sum\limits_{\substack{u\in V(T) \setminus L(T)\\v\in V(T)\setminus \set{\rho_T}\\y(u)>y(v)}} a(u,v) \cdot z_{u,v}\]

\paragraph*{Fixed species tree.}

The ILP formulation can also be used for the \FTT problem where an order of the species is part of the input.
In terms of the variables and constraints above, 
this means that for all $s \in L(S)$ the variables $J^l_s$ and $J^r_s$ become constants.

\section{Example Drawings}
\label{app:examples}
We give drawings for one instance of each of the three data sets used in our experiments.

\begin{figure}[htb]
  \centering
    \begin{subfigure}[t]{\linewidth}
		\centering
 		\includegraphics[width=0.9\linewidth]{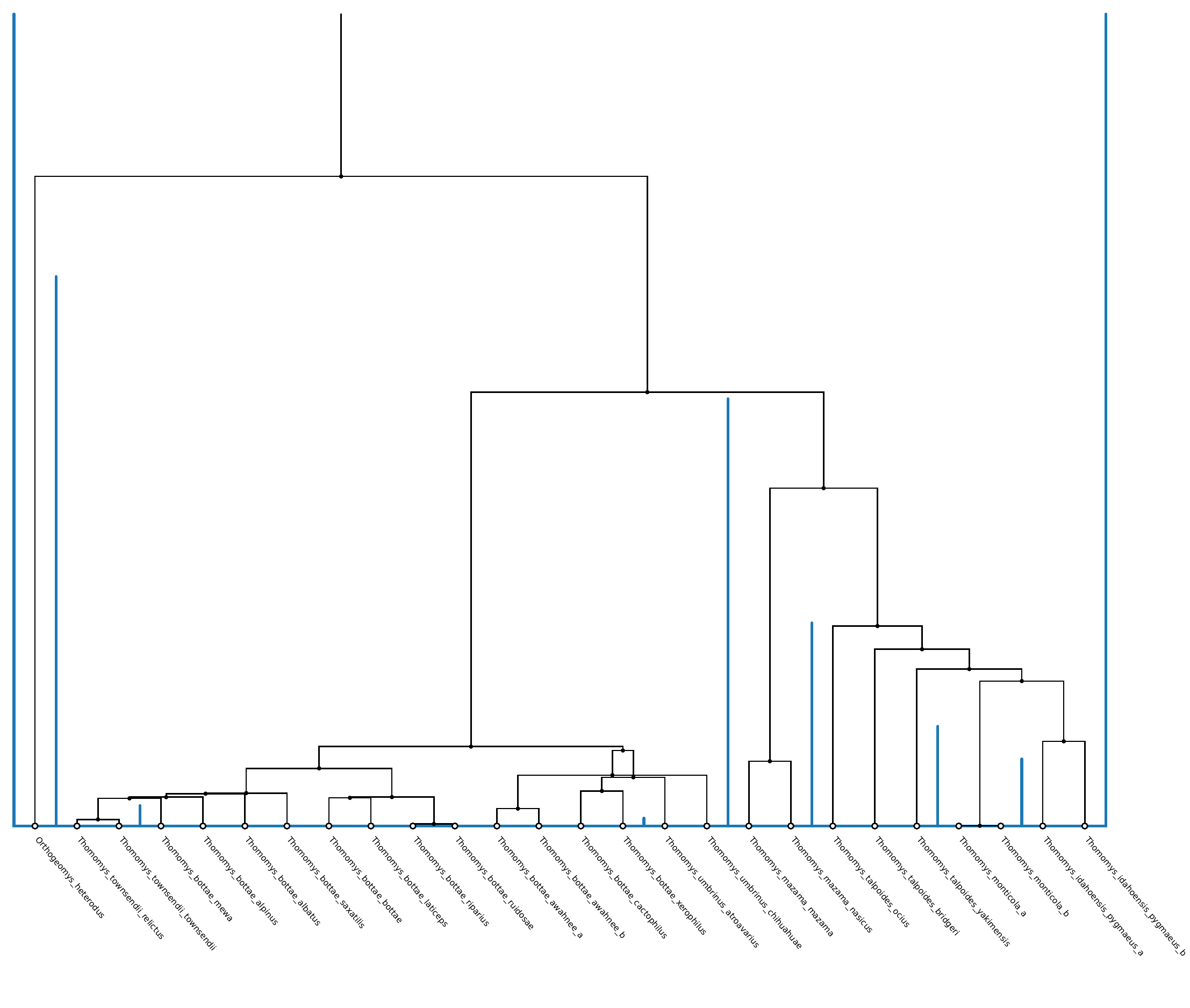} 
		\caption{Rectangular drawing style.}
	\end{subfigure}
	\begin{subfigure}[t]{\linewidth}
		\centering
 		\includegraphics[width=0.9\linewidth]{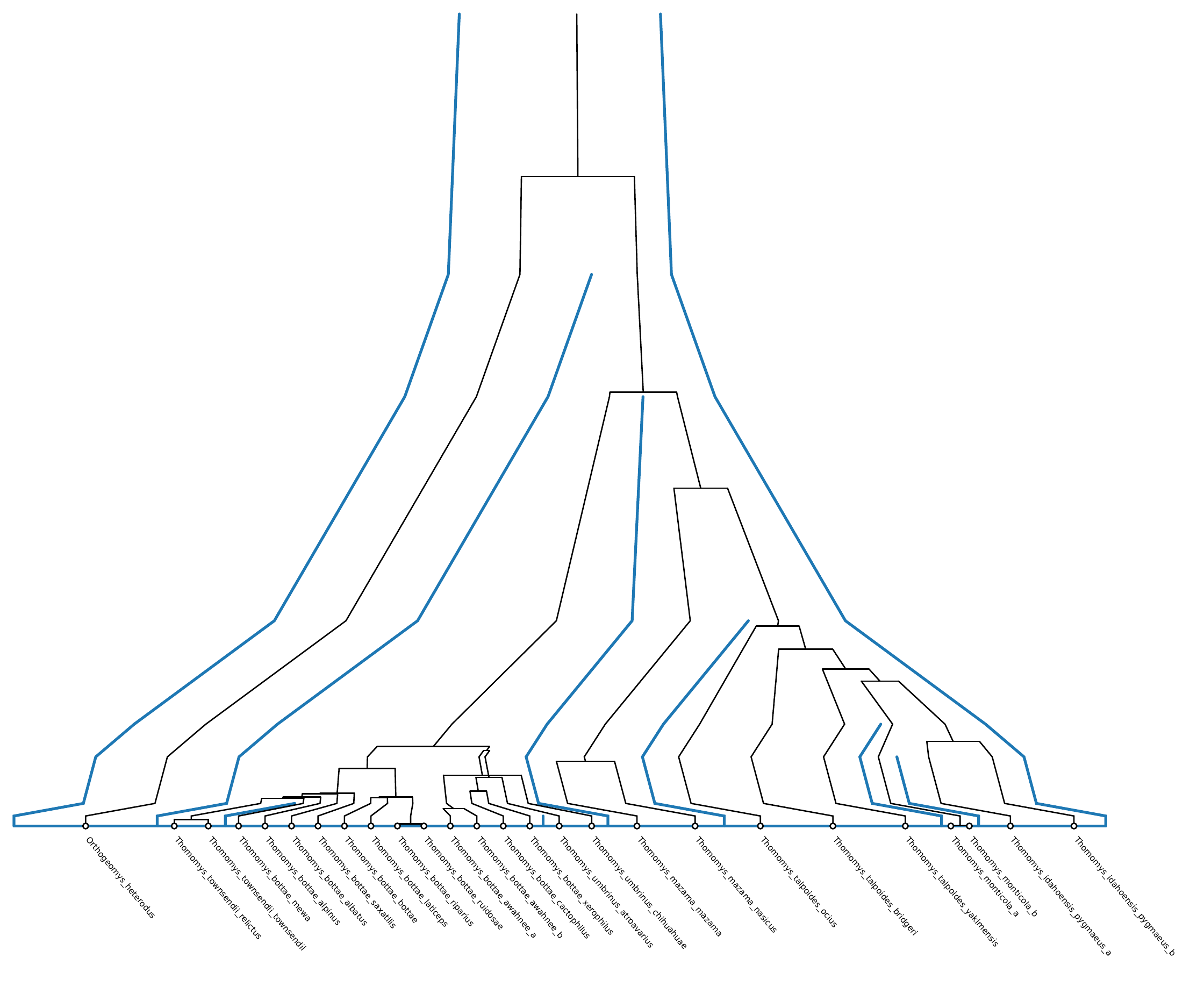}
		\caption{Proportional drawing style.}
	\end{subfigure}
  \caption{Drawings of a Gopher instance.}
  \label{fig:gopher}
\end{figure}

\begin{figure}[htb]
  \centering
    \begin{subfigure}[t]{\linewidth}
		\centering
 		\includegraphics[width=0.9\linewidth]{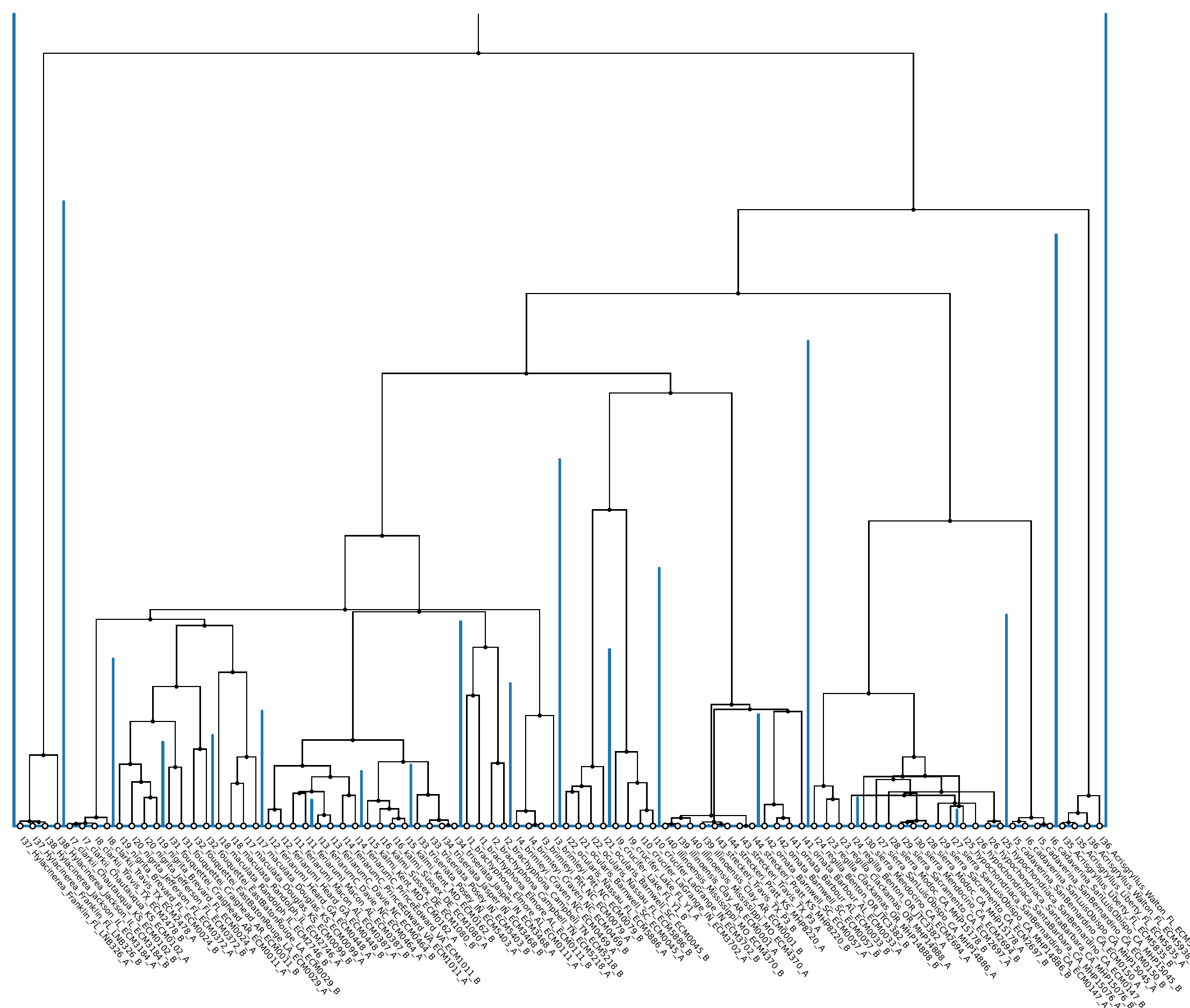} 
		\caption{Rectangular drawing style.}
	\end{subfigure}
	\begin{subfigure}[t]{\linewidth}
		\centering
 		\includegraphics[width=0.9\linewidth]{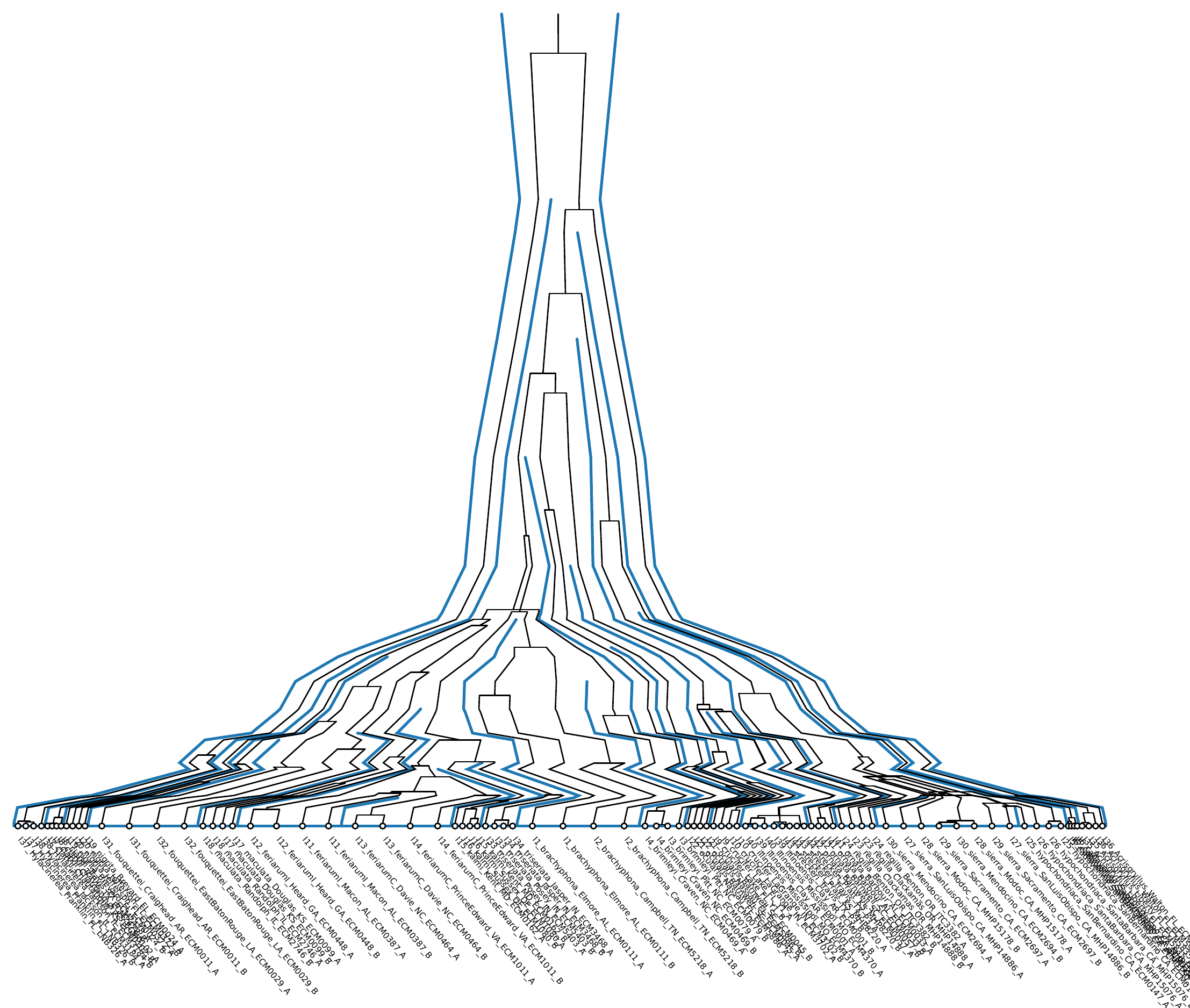}
		\caption{Proportional drawing style.}
	\end{subfigure}
  \caption{Drawings of a Barrow instance.}
  \label{fig:gopher}
\end{figure}

\begin{figure}[htb]
  \centering
    \begin{subfigure}[t]{\linewidth}
		\centering
 		\includegraphics[width=0.9\linewidth]{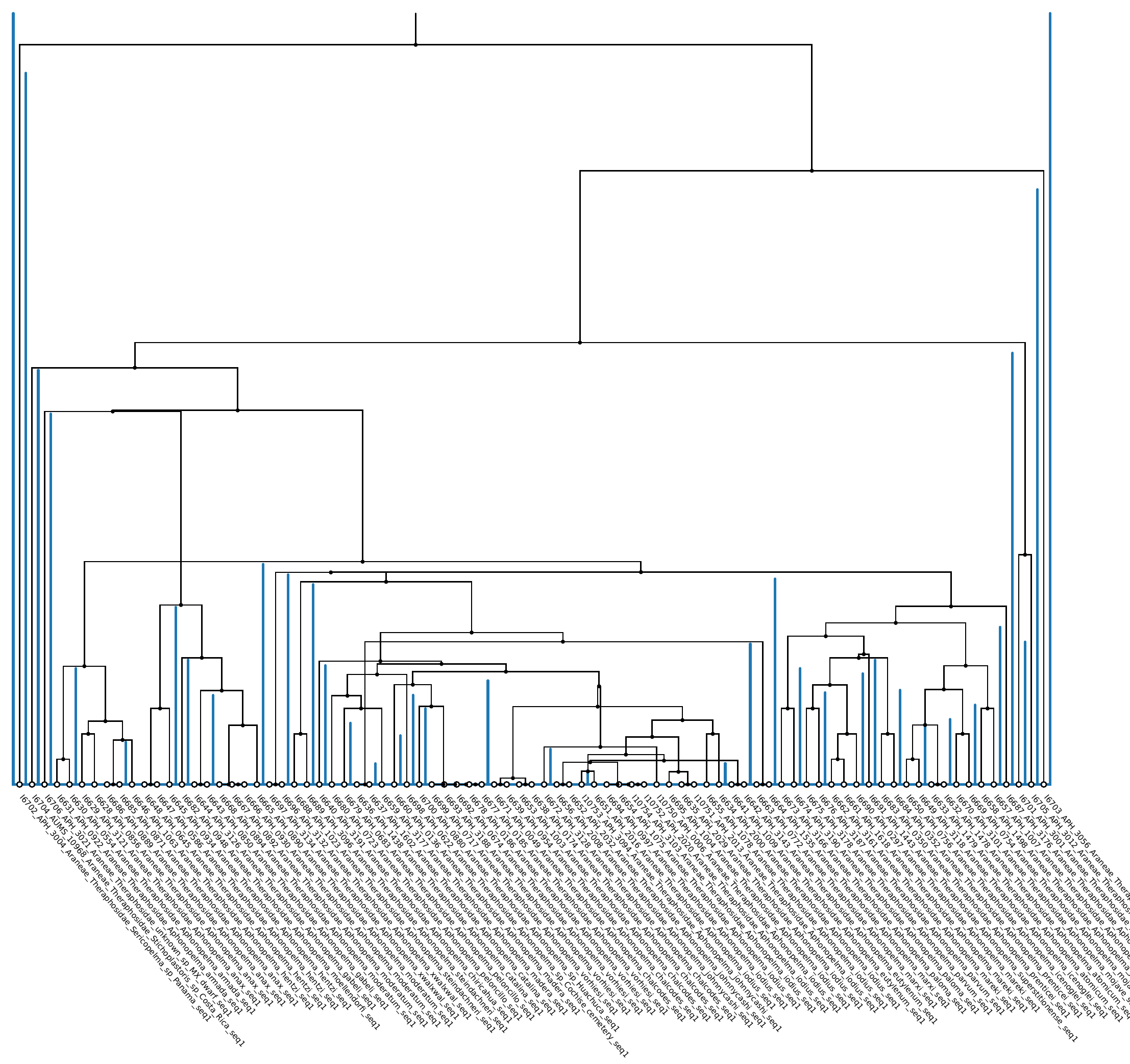} 
		\caption{Rectangular drawing style.}
	\end{subfigure}
	\begin{subfigure}[t]{\linewidth}
		\centering
 		\includegraphics[width=0.9\linewidth]{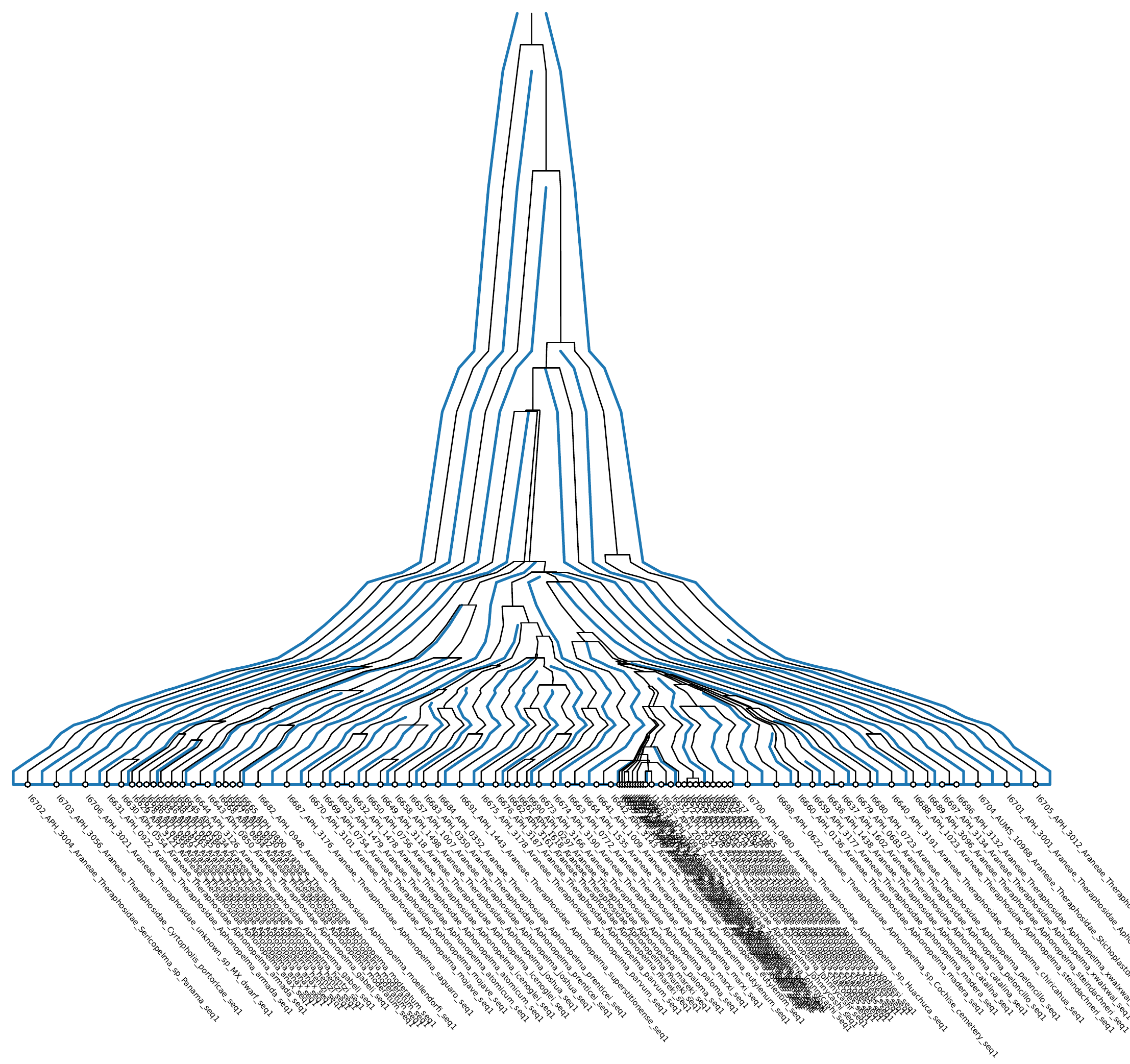}
		\caption{Proportional drawing style.}
	\end{subfigure}
  \caption{Drawings of a Hamilton instance.}
  \label{fig:gopher}
\end{figure}

\end{document}